\documentclass[12pt, a4paper]{article}
\usepackage{amsmath}
\usepackage{amsthm}
\usepackage{amsfonts}
\usepackage{multicol}
\usepackage[dvips]{graphicx}
\usepackage{fixltx2e}
\usepackage{tikz-cd}
\usepackage{IEEEtrantools}

\newcommand{\danger}[1]{\textbf{#1}}

\input{epsf}             

\newtheorem{definition}{Definition}
\newtheorem{theorem}{Theorem}
\newtheorem{lemma}{Lemma}

\begin{document}

\title{\danger{Quantum Geometry II : \\
\small{The Mathematics of Loop Quantum Gravity \\} 
\normalsize{Three dimensional quantum gravity}}}
\author{\centerline{\danger{J. Manuel Garc\'\i a-Islas \footnote{
e-mail: jmgislas@gmail.com}}}  \\
}

\maketitle

\begin{abstract}
Loop quantum gravity is a physical theory which aims at unifying general relativity and quantum mechanics. 
It takes general relativity very seriously and modifies it via a  quantisation. General relativity describes gravity 
in terms of geometry. Therefore, quantising such theory must be equivalent to quantising geometry and that is
what loop quantum gravity does. This sounds like a mathematical task as well. 
This is why in this paper we will present the mathematics of loop quantum gravity. We will
do it from a mathematician point of view. This paper is intended to be an introduction to loop quantum gravity for postgraduate students of physics 
and mathematics.
In this work we will restrict ourselves to the three dimensional case.           
\end{abstract}

\bigskip

\bigskip

Key Words: Quantum Gravity, Quantum Geometry. 

PACS Number: 04.60, 04.60.Pp 

\bigskip 

\bigskip

\bigskip

\bigskip 

\bigskip

\bigskip

\bigskip

\newpage

\section{Introduction}

It is very well known in physics that 
classical physical theories have a quantum version. Classical Mechanics and Quantum Mechanics,
Classical Electrodynamics and Quantum Electrodynamics, Special Relativity and Quantum Field Theory. 

The question then is: what is the
quantum version of general relativity? 
There is no agreement on the answer to this question.

Loop quantum gravity \cite{r}, \cite{rv}, \cite{t} \cite{gp}, 
is a theory which describes the quantum version of general relativity. Since general relativity is a geometrical theory,
quantising such a theory is equivalent to thinking of quantising classical geometry.\footnote{
Mathematically speaking,  
quantum geometry may also refer to a subject known as noncommutative geometry \cite{ac}. But for us it will refer
to loop quantum gravity which is also a quantum geometry.}
This is what loop quantum gravity does in its own original way.   However, unfortunately mathematicians are not
aware of loop quantum gravity and it may be due to the fact that it is a theory
mostly invented by physicists. The truth is that loop quantum gravity is such a beautiful theory with
very precise, rigorous mathematical background which exploits many different subjects of outstanding
level. 

Our intention and our project is to write papers introducing loop quantum gravity to postgraduate students in physics and mathematics. We
write it from the point of view of a mathematician. We hope that more mathematicians get to know
and become interested in loop quantum gravity. 

However, the project may also be considered from the loop quantum gravity community itself as 
writing about the mathematics of loop quantum gravity. 

Loop quantum gravity is a very extensive field of study and there are numerous things to deal with. Our project is to 
write about the mathematics of loop quantum gravity and therefore write many works on this subject by dealing with a particular theme 
in each work. We have already started this project describing quantum polyhedra in \cite{gi}.

In this work we will deal with three dimensional quantum gravity which has been extensively studied in the literature and there are also
many things in this particular theme to write about. We concentrate on the essentials. Our presentation of the mathematics
of loop quantum gravity is completely original and described in our own way. In this sense this is an original work and an
original project.

In section 2 and 3 we present the mathematical background which is also useful for many other themes in loop quantum gravity. 
In section 4 and 5 we introduce three dimensional loop quantum gravity.

\section{Topological groups and representations}

In this section we review some mathematical background related to Lie groups and linear representations which is used in loop quantum gravity.
This section is based on \cite{yks}, \cite{bh}. 
  
\begin{definition}
A topological group $G$, is a group which is also a topological space such that the product and inverse operations are continuous. It is said to be compact 
if it is a compact space.  
\end{definition}
We will only consider compact topological groups.

\bigskip

There exists a unique measure on $G$ with the following properties.

\bigskip

1) $\int_{G} f(g) \ dg = \int_{G} f(g h) \ dg$ for every $h \in G$.

\bigskip

2) $\int_{G} f(g) \ dg = \int_{G} f(h g) \ dg$ for every $h \in G$.

\bigskip

3) $\int_{G} f(g) \ dg = 1$.

\bigskip

The measure $dg$ is called invariant measure or the Haar measure of $G$.\footnote{Since we are considering compact Lie groups only, it is known that right invariant measures
are also left invariant and that a normalisation condition such as 3) can be chosen. We will not be concerned with this here.}

\begin{definition}
A linear representation of $G$ in a Hilbert space $H$ is a group homomorphism $\rho : G \longrightarrow GL(H)$, such that $\rho(g) x$ is a continuous map
for every $x \in H$.
\end{definition}

Therefore, we have that for $g$ and $g' \in G$,  $\rho(g g') = \rho(g) \rho(g')$, $\rho(g^{-1}) = \rho(g)^{-1} $ and $\rho(e) = Id_{H}$, where $e$ is the identity 
of the group and $Id_{H}$ is the identity map in $H$.

The dimension of $H$ is called the dimension of the representation and a representation will be denoted by the pair $(H, \rho)$.

\begin{definition}
A linear representation of $G$ in a Hilbert space $H$ is a unitary representation if $\rho(g)$ is a unitary operator for every $g \in G$.
Equivalently, $\rho$ is unitary if $< \rho(g) x , \rho(g) y > = < x , y >$ for every $g \in G$ and every $x, y \in H$.\footnote{$< , >$ is the scalar product in H
which we assume to be conjugate linear in the first entry and linear in the second} 
\end{definition}

If $(e_i)$ is a basis of $H$ we have the matrix coefficients of the representation $(H, \rho)$ given by

\begin{equation}
\rho_{ij}(g) = < e_i , \rho(g) e_j > \nonumber
\end{equation}

\begin{definition}
An intertwining operator between two unitary representations $(H_1 , \rho_1)$ and $(H_2, \rho_2)$ is a continuos linear map $T: H_1 \longrightarrow H_2$ such that 
for every $g \in G$ the following diagram commutes
\end{definition}

\begin{tikzcd}
H_1 \arrow[r, "T"]\arrow[d,  "\rho_1(g)"red]
& H_2 \arrow[d, "\rho_2(g)" ] \\
H_1 \arrow[r,  "T" red]
& H_2
\end{tikzcd}

\begin{definition}
The representations $(H_1 , \rho_1)$ and $(H_2, \rho_2)$ are said to be equivalent if there exists a bijective $T: H_1 \longrightarrow H_2$ 
intertwining operator. 
\end{definition}
This implies that we have an equivalence class of a representation.

\begin{definition}
A vector subspace $W \subset H$ is said to be invariant under $(H , \rho)$, if $\forall g \in G$, $\rho(g) W \subset W$.
\end{definition}

\begin{definition}
A representation $(H , \rho)$ is said to be irreducible if $H \neq \{ 0 \}$ and if $\{ 0 \}$ and $H$ are the only 
vector subspaces of $H$ which are invariant under $\rho$.
\end{definition}

\begin{definition}
Let $(H_1 , \rho_1)$ and $(H_2 , \rho_2)$ be representations of $G$. The direct sum is the representation $(H_1 \oplus H_2 , \rho_1 \oplus \rho_2)$
so that $\forall g \in G$ and $x_1 \in H_1, x_2 \in H_2$.
\end{definition}

\begin{equation}
( \rho_1 \oplus \rho_2)(g)(x_1 , x_2) = (\rho_1(g) x_1 ,  \rho_2(g)x_2)\nonumber
\end{equation}
The matrix representation of a direct sum is given by 

\begin{equation*}
    \left(
      \begin{array}{ccc}
        \rho_1(g) & 0   \\
        0 & \rho_2(g)  
        \end{array} \right)
  \end{equation*}
Analogously the direct sum of an arbitrary finite number of representations is defined.

\begin{definition}
Let $(H_1 , \rho_1)$ and $(H_2 , \rho_2)$ be representations of $G$. The tensor product is the representation $(H_1 \otimes H_2 , \rho_1 \otimes \rho_2)$
so that $\forall g \in G$ and $x_1 \in H_1, x_2 \in H_2$.
\end{definition}

\begin{equation}
( \rho_1 \otimes \rho_2)(g)(x_1 \otimes x_2) = \rho_1(g) x_1 \otimes  \rho_2(g)x_2 \nonumber
\end{equation}
The matrix representation of a tensor product is given by the Kronecker tensor product of the matrices

\begin{equation*}
     \left(
      \begin{array}{ccc}
        \rho_{11} \rho_2(g) & \cdots & \rho_{1n} \rho_2(g) \\
        
        \vdots & \vdots & \ddots \\
        \rho_{n1} \rho_2(g) & \cdots & \rho_{nn} \rho_2(g)
      \end{array} \right)
  \end{equation*}

\begin{definition}
A representation $(H , \rho)$ is said to be reducible  if it is the direct sum of irreducible representations
$(W_1 , \rho_1), (W_2 , \rho_2),..., (W_k , \rho_k)$; this is expressed 
$( H , \rho)  = ( W_1 \oplus W_2 \oplus \cdots \oplus W_k , \rho_1 \oplus \rho_2 \oplus \cdots \oplus \rho_k )$
\end{definition}
or simply

\begin{equation}
 ( H , \rho) = ( \oplus_{i=1}^{k} W_i , \oplus_{i=1}^{k} \rho_i ) \nonumber
\end{equation}

\begin{theorem}
Every unitary representation $(H, \rho)$ of $G$ is reducible.
\end{theorem}

\begin{proof}
If the representation is irreducible, then we are finished. Suppose then that the representation $(H , \rho)$ is unitary and not irreducible; then there
must be an invariant subspace $W \subset H$. Since there is a scalar product we can consider the complement of $W$ in $H$ denoted $W^{\perp}$
such that $H = W \oplus W^{\perp}$. It can be shown that $W^{\perp}$ is also an invariant subspace. By induction the theorem follows. 
\end{proof}

\begin{lemma}\danger{( Schur )} Let $(H_1 , \rho_1)$ and $(H_2 , \rho_2)$ be unitary irreducible representations of $G$. And let $T: H_1 \longrightarrow H_2$
be an intertwining operator. Then $T = 0$ or $T$ is an isomorphism. 
\end{lemma}

Consider the vector space $\mathbb{C}^{G} =\{ f : G \longrightarrow \mathbb{C} \}$ of complex valued functions on $G$, and define the scalar product

\begin{equation}
 < f_1 \mid f_2 >  = \int_{G} \overline{f_1 (g)}  f_2 (g) \ dg \nonumber
\end{equation}
where $dg$ is the Haar measure. By completing this space we obtain the Hilbert space of square integrable functions $L^{2}(G)$.

\begin{theorem}
Let $(H_1 , \rho_1)$ and $(H_2 , \rho_2)$ be irreducible unitary representations of $G$. Let $x_1, y_1 \in H_1$ and $x_2, y_2 \in H_2$. Then
\begin{equation}
 < \rho_{x_1 y_1}(g) \mid \rho_{x_2 y_2}(g) > = 0 \nonumber
\end{equation}
if the representations are not equivalent, or
\begin{equation}
 < \rho_{x_1 y_1}(g) \mid \rho_{x_2 y_2}(g) > = \frac{1}{n} < x_2 , x_1 > < y_1 , y_2 >\nonumber
\end{equation}
if the representations are equivalent and where $n$ is the dimension of $H_1 \sim H_2$.
\end{theorem}
In particular, if we have orthonormal basis

\begin{equation}
< \rho_{i j}(g) \mid \rho_{k l}(g)) > = 0 \nonumber
\end{equation}
if the representations are not equivalent, or

\begin{equation}
< \rho_{i j}(g) \mid \rho_{k l}(g) > = \frac{1}{n} \delta_{i k} \delta_{j l}  \nonumber
\end{equation}
if the representations are equivalent, where $\delta_{ij}$ is the Kronecker delta. 

The set of equivalence classes of irreducible representations of $G$ is denoted by $\hat{G}$.

\begin{theorem}\danger{Peter - Weyl.}
Let $f \in L^{2}(G)$. Then
\begin{equation}
f(g) = \sum_{\alpha \in \bar{G}} \sum_{i , j = 1}^{dim \ \rho^{\alpha}} c_{ij}^{\alpha} \ \rho_{ij}^{\alpha}(g) \nonumber
\end{equation}
where $\alpha$ runs over the different irreducible unitary representations, $\rho_{ij}^{\alpha}$ are the matrix coefficients of the representations in
orthonormal basis.  
\end{theorem}
Therefore, the coefficients are given by 

\begin{equation}
c_{ij}^{\alpha} = dim \ \rho^{\alpha} \ \int_{G} \overline{\rho_{i j}^{\alpha}(g)} f(g) \ dg \nonumber
\end{equation}

\subsection{SU(2)}

In this part, we will summarise some important aspects of the irreducible representations of the group
$G = SU(2)$ which are used in LQG. $SU(2)$ consist of the matrices given by

\begin{equation*}
\left(
      \begin{array}{ccc}
        a & b   \\
        - \overline{b} & \overline{a}  
        \end{array} \right)
  \end{equation*}
where $a, b \in \mathbb{C}$ and $\mid a \mid^2 + \mid b \mid^2 = 1$. It can be easily seen that $SU(2)$ is diffeomorphic to the three dimensional sphere $S^3$. 

$SU(2)$ acts on the space of complex valued functions on $\mathbb{C}^2$ as follows: we have that for $g \in SU(2)$, $g^{-1} = {}^{t}_{}\overline{g}$; the action
is given by

\begin{equation}
\rho(g) f = f \circ g^{-1} \nonumber
\end{equation}
so that 

\begin{equation}
\rho(g) f(z_1 , z_2)  = f( \overline{a}z_1 - b z_2 , \overline{b}z_1 + az_2) \nonumber
\end{equation}
Now we restrict the space of complex valued functions on $\mathbb{C}^2$ to homogeneous polynomials and denote by $H^j$ the vector space of homogeneous polynomials
of degree $2j$ where $j = \frac{1}{2} \mathbb{N}$. The dimension of $H^j$ is $2j +1$.\footnote{This view is common to physicists and we will use it.
Mathematicians are more used to think of the space $H^n$ as the space of complex homogeneous polynomials on $\mathbb{C}^2$ of degree $n =2j, n \in \mathbb{N}$,
which has dimension $n+1$.} 

This space has basis given and denoted by

\begin{equation}
f_{m}^{j} (z_1 , z_2) = {z_1}^{j+m} \ {z_2}^{j-m} \nonumber
\end{equation}
where $-j \leq m \leq j$; $m$ is integer if $j$ is integer and half-integer if $j$ is half-integer. When restricted to the space $H^j$ the action $\rho$ gives a representation which we denote by $(H^j ,\rho^j )$.

The Lie algebra $\mathfrak{su}(2)$ of $SU(2)$ is a vector space of dimension three of traceless anti-Hermitian matrices with basis
given by

\begin{equation*}
\xi_1 = \frac{1}{2} \left(
      \begin{array}{ccc}
        0 & i   \\
        i & 0  
        \end{array} \right) \qquad
\xi_2 = \frac{1}{2} \left(
      \begin{array}{ccc}
        0 & -1   \\
        1 & 0  
        \end{array} \right) \qquad
\xi_3 = \frac{1}{2} \left(
      \begin{array}{ccc}
        i & 0   \\
        0 & -i  
        \end{array} \right)                 
  \end{equation*}
which satisfy the commutation relations $[\xi_1 , \xi_2] = \xi_3 , [\xi_2 , \xi_3] = \xi_1 , [\xi_3 , \xi_1] = \xi_2 $.  

Physicists in LQG prefer to use the Hermitian matrices given by $J_k = i \xi_k$ for $k=1,2,3$ where the commutation relations are given by 
$[J_1 , J_2] = i J_3 , [J_2 , J_3] = i J_1 , [J_3 , J_1] = i J_2 $.

Moreover, we consider the matrices $J_3 , J_+ = J_1 + i J_2 , J_- = J_1 - i J_2$ which have commutation relations 
$[J_+ , J_-] = 2 J_3 , [J_3 , J_+] =  J_+ , [J_3 , J_-] = - J_- $. 
It is known that the differentials $D$ of the representations $\rho^j$ are representations of the Lie algebra and therefore
give rise to operators $( D \rho^j ) J_3 ,  (D \rho^j ) J_+ ,  (D \rho^j ) J_-$
that act on the space $H^j$ as follows\footnote{The differential of the representation $(E, \rho)$ of a Lie Group $G$ gives a Lie Algebra representation.
Studying Lie Algebra representations is equivalent to studying the representations of the Lie Group. See \cite{yks} for details.}; rename the basis $f_{m}^{j} $ of $H^j$ with notation used in physics

\begin{equation}
\mid j , m > = \frac{1}{\sqrt{(j-m)! (j+m)!}}f_{m}^{j}  \nonumber
\end{equation}
then 

\begin{align}
&( D \rho^j ) J_3 \ \mid j , m > = m \ \mid j , m >  \nonumber \\
&( D \rho^j ) J_+ \ \mid j , m > = \sqrt{(j-m) (j+m +1)} \ \mid j , m +1 > \nonumber  \\
& ( D \rho^j ) J_- \ \mid j , m > =  \sqrt{(j+m) (j-m +1)} \ \mid j , m -1 > \nonumber
\end{align}
It is a well known proven fact that each representation $(H^j , \rho^j)$ of $SU(2)$ is unitary and irreducible with the fact that the basis
$ \mid j , m >$ is orthonormal. 

The tensor product of representations is used thoroughly in LQG. The tensor product of two irreducible representations is reducible and therefore 
is the direct sum of irreducible representations which for $SU(2)$ is known as the Clebsch-Gordan decomposition

\begin{equation}
 ( H^{j_1} \otimes H^{j_2}, \rho^{j_1} \otimes \rho^{j_2} )  = ( H^{\mid j_2 -j_1 \mid} \oplus H^{\mid j_2 -j_1 \mid +1} \oplus \cdots \oplus 
 H^{j_1 + j_2} , \rho^{\mid j_2 -j_1 \mid} \oplus \cdots \oplus \rho^{j_1 + j_2} ) \nonumber
\end{equation}
This decomposition of the tensor product of two unitary irreducible representations implies that we can consider
two orthonormal basis, one given by

\begin{equation}
\mid j_1 , m_1 > \otimes \mid j_2 , m_2 > \qquad \text{where} \qquad - j_1 \leq m_1 \leq j_1, \qquad - j_2 \leq m_2 \leq j_2 \nonumber
\end{equation}
And the second one given by

\begin{equation}
\mid J , M >   \qquad \text{where} \qquad \mid j_1 - j_2 \mid \leq J \leq j_1+ j_2 \qquad \text{and} \qquad - J \leq M \leq J \nonumber
\end{equation}
Since the tensor product and the direct sum decomposition are isomorphic spaces 
we can change basis and write down one of the basis vectors as a linear combination of the others

\begin{equation}
\mid J , M >  = \sum_{m_1 , m_2} C(J , M , j_1 , m_1 , j_2 , m_2) \ \mid j_1 , m_1 > \otimes \mid j_2 , m_2 > \nonumber
\end{equation}  
where the coefficients $C(J , M , j_1 , m_1 , j_2 , m_2)$ are known as Clebsch-Gordan coefficients.
It is more common in physics to use Wigner notation of these coefficients because of the symmetries which are easier to visualise. Let

\begin{equation*}
\left(
      \begin{array}{ccc}
        j_1 & j_2 & J   \\
        m_1 & m_2 & - M  
        \end{array} \right) = \frac{(-1)^{j_1 - j_2 + M}}{\sqrt{2J+1}} \ C(J , M , j_1 , m_1 , j_2 , m_2)
  \end{equation*}
Then we can write the change of basis formula as

\begin{equation*}
\mid J , M >  = \sum_{m_1 , m_2} (-1)^{j_2 - j_1 - M} \sqrt{2J+1}
\left(
      \begin{array}{ccc}
        j_1 & j_2 & J   \\
        m_1 & m_2 & - M  
        \end{array} \right) \
\mid j_1 , m_1 > \otimes \mid j_2 , m_2 > \nonumber
\end{equation*}  
It is also useful to write down the change of basis the other way around. This is given by 

\begin{equation*}
\mid j_1 , m_1 > \otimes \mid j_2 , m_2 > 
= \sum_{J , M} (-1)^{j_2 - j_1 - M} \sqrt{2J+1}
\left(
      \begin{array}{ccc}
        j_1 & j_2 & J   \\
        m_1 & m_2 & - M  
        \end{array} \right) \
\mid J , M >  \nonumber
\end{equation*}  
The symmetry properties of the Wigner symbol are can be found in the literature.

\section{Graphs and Triangulations}
 
This section is based on \cite{drs}.

\begin{definition}
A point configuration in Euclidean space $\mathbb{R}^n$ is a finite collection of points $\Pi = \{ p_1, p_2, ....p_k \}$  
\end{definition}

\begin{definition}
A convex hull of $\Pi$ is the intersection of all convex sets which contain the points of $\Pi$. It is denoted by $\text{conv}(\Pi)$.
\end{definition}

\begin{definition}
A d-simplex is the convex hull of $d+1$ affinely independent points in $\mathbb{R}^n$. We must have $n \geq d$. 
\end{definition}

\begin{definition}
A j-face of a d-simplex is the convex hull of $j+1$ of its vertices. 
\end{definition}

\begin{definition}
A triangulation of a point configuration $\Pi \in \mathbb{R}^n$ is a collection of n-simplices all of whose vertices are points of $\Pi$ such that
the union of all these simpleces equals $\text{conv}(\Pi)$ and any pair of these simpleces intersects in a common face. 
\end{definition}
It is allowed that the intersection of these pair of simpleces is empty.

\begin{definition}
A graph is a pair $\Gamma = (V, E)$ of sets, such that $E \subset [V]^2$. 
\end{definition}
$E \subset [V]^2$ means that the elements of $E$ are two-element subsets of $V$. Is is always assumed that $V \cap E = \emptyset$. 
 The elements of $V$ will be called vertices and the elements of $E$ will be called edges.

The vertex set of the graph $\Gamma$, will be denoted as $V(\Gamma)$ and its vertex set by $E(\Gamma)$.

\begin{definition}
The number of vertices of the graph $\Gamma$ is its order and is denoted $\mid \Gamma \mid$.
\end{definition}

 \begin{definition}
A vertex is incident to edge $e$ if $v \in e$. It is also said that $e$ is an edge at $v$. 
The two vertices incident to edge $e$ are called its endvertices. 
\end{definition}
 The set of all edges at a vertex $v$ is denoted by $E(v)$.

\begin{definition}
The number of edges $E(v)$ at vertex $v$ is called the degree of vertex $v$ and it is denoted by $d_{\Gamma}(v)$ or simply by $d(v)$. 
\end{definition}

\begin{definition}
Given a triangulation of a point configuration $\Pi \in \mathbb{R}^n$ we associate a graph called
the dual graph; this graph is drawn as follows:  draw a node for each n-simplex and an 
edge when two n-simpleces intersect in a (n-1)-face.
\end{definition}

Consider a triangulation of a point configuration in $\mathbb{R}^3$ and its dual graph. The dual graph is called a 2-complex in LQG. Observe that the 
dual graph associated to the boundary simplicial triangulation is given by a trivalent graph $\Gamma=(V , E)$.

 \section{Three dimensional LQG}
 
 We now start with the three dimensional study of quantum gravity as it is understood in LQG. This section is our main section and we 
 try to be as mathematical as possible as in our previous sections. This section is entirely inspired on \cite{rv}. 
A mathematical description of loop quantum gravity can be found in \cite{jb}, and reference \cite{t} is a very advanced
book on loop quantum gravity which is very mathematical. A very nice and useful introduction to the subject which uses more physical notation 
is \cite{md}.

\subsection{Spin Networks} 

In LQG there is a Hilbert space which is interpreted as the quantum space of states. 
Consider a trivalent directed graph $\Gamma = (V, E)$  with $\mid V \mid$ number of vertices  and $\mid E \mid$ number of edges. 
The Hilbert space of three dimensional LQG is a subspace of the Hilbert space $L^2(SU(2)) \times L^2(SU(2)) \cdots \times L^2(SU(2))$ of $\mid E \mid$ products
which we denote by $L^2(SU(2)^{\mid E \mid})$. This subspace is called the invariant space and it is the Hilbert space of the theory. This
space is given as follows. 

Consider a function $c: E \rightarrow \widehat{SU(2)}$ which goes from the set of edges to the
set of equivalence classes of irreducible representations of $SU(2)$ with the following properties. This function $c$ is called 
labelling function.  

So for each directed edge $e$ we have $c(e) = ( H^{j} , \rho^{j} )$. Denote the matrix coefficients of this representation by 

\begin{equation}
\rho_{mn}^{j} (g) = \ < e_m , \rho^{j}(g) e_n > \nonumber
\end{equation} 
We can think of this in a graphical way, where the directed edge is labelled $j$ and its vertices are labelled $m$ and $n$ and the edge is directed
from vertex labelled $n$ to vertex labelled $m$.

For each vertex $v$ we have three directed edges $e_1, e_2, e_3$ incident to $v$ and suppose that 
$c(e_1) = ( H^{j_1} , \rho^{j_1} )$ and $c(e_2)= ( H^{j_2} , \rho^{j_2} )$ and $c(e_3)= ( H^{j_3} , \rho^{j_3} )$. 

Assign to the vertex $v$ the Wigner coefficients \footnote{The Wigner coefficient $\widetilde{C}(j_1 , m_1 , j_2 , m_2 , j_3 , m_3)$ 
in a graphical way is thought as a trivalent vertex with its edges directed outwards and labelled $j_1, j_2, j_3$. If we had 
$\widetilde{C}(j_1 , m_1 , j_2 , m_2 , j_3 ,  - m_3)$  then it is thought as a trivalent vertex with
two of its edges directed outwards labelled $j_1, j_2$ and one inwards labelled $j_3$ Analogously if any of the indices $m_i$ changes sign.}

\begin{equation*}
v \mapsto \left(
      \begin{array}{ccc}
        j_1 & j_2 & j_3   \\
        m_1 & m_2 & m_3  
        \end{array} \right) : = \widetilde{C}(j_1 , m_1 , j_2 , m_2 , j_3 , m_3) 
 \end{equation*}
This is equivalent to assigning to the vertex $v$ an intertwining operator given by $T : \mathbb{C} \longrightarrow H^{j_1}  \otimes H^{j_2} \otimes H^{j_3} $. 
Although the notation of the Wigner coefficients is the $2 \times 3$ matrix, we will use the $\widetilde{C}(j_1 , m_1 , j_2 , m_2 , j_3 , m_3)$ notation  
of our own.

Now we are ready to present an element of the invariant Hilbert space of three dimensional LQG. 

By the Peter-Weyl theorem if $f \in L^2(SU(2)^{\mid E \mid})$ it must be of the form

\begin{equation}
f(g_1 , g_2 , \cdots g_{\mid E \mid}) = \sum_{j_1 \cdots j_{\mid E \mid}} \sum_{m , n = 1}^{dim \ \rho^{j}} 
c_{m_1 m_2 \cdots m_{\mid E \mid} n_1 n_2 \cdots n_{\mid E \mid}}^{j_1 j_2 \cdots j_{\mid E \mid}} \ 
\rho_{m_1 n_1}^{j_1}(g_1) \cdots \rho_{m_{\mid E \mid} n_{\mid E \mid}}^{j_{\mid E \mid}}(g_{\mid E \mid}) \nonumber
\end{equation}
According to \cite{rv}, the invariant Hilbert space is spanned by the functions $\Psi \in L^2(SU(2)^{\mid E \mid})$ given by
\footnote{Recall that we have considered a trivalent graph $\Gamma = (V, E)$  with $\mid V \mid$ number of vertices  and $\mid E \mid$ number of edges.} 

\begin{IEEEeqnarray}{rCl}
\Psi(g_1 , g_2 , \cdots g_{\mid E \mid}) = & \sum_{m_1 m_2 \cdots n_{\mid E \mid-1} n_{\mid E \mid}} & \
\widetilde{C_1}(j_1, m_1, j_2 , m_2 , j_3 , m_3)  \times \cdots \times \nonumber \\
& & \widetilde{C_{\mid V \mid}}( j_{\mid E \mid-2} ,  n_{\mid E \mid-2} ,  j_{\mid E \mid-1} ,  n_{\mid E \mid-1} ,  j_{\mid E \mid} ,  n_{\mid E \mid} ) \times \nonumber \\ 
&  & \rho_{m_1 n_1}^{j_1}(g_1) \cdots \cdots \rho_{m_{\mid E \mid} n_{\mid E \mid}}^{j_{\mid E \mid}}(g_{\mid E \mid}) \nonumber
\end{IEEEeqnarray}
These functions are known as the quantum state functions. Therefore the invariant Hilbert space is given by linear combinations of these functions.
This space is called Quantum Space. 

Let us show with examples how these generator functions $\Psi$ are given.  

\subsubsection{Example}

\bigskip

\bigskip

\danger{The theta graph}

\bigskip

The theta graph $\Theta = (V, E)$ has two vertices $v_1, v_2$ and three edges $e_1, e_2 , e_3$. If $c(e_1) = ( H^{j_1} , \rho^{j_1} )$, 
$c(e_2) = ( H^{j_2} , \rho^{j_2} )$ then recall that $j_3$ must be chosen between $\{  \mid j_2 - j_1\mid , \mid j_2 - j_1\mid +1 , \cdots j_1 + j_2 \ \}$. 
In this case the state function is given by 

\begin{IEEEeqnarray}{rCl}
\Psi (g_1 , g_2 , g_3) = & \sum_{m_1 m_2 m_3}  \ \sum_{n_1 n_2 n_3} & \
\widetilde{C}(j_1, m_1, j_2 , m_2 , j_3 , m_3) \ \widetilde{C}(j_1, - n_1, j_2 , - n_2 , j_3 , - n_3) \nonumber \\
&& \rho_{m_1 n_1}^{j_1}(g_1) \ \rho_{m_2 n_2}^{j_2}(g_2) \ \rho_{m_3 n_3}^{j_3}(g_3) \nonumber
\end{IEEEeqnarray}
where the sum takes the values $- j_1 \leq m_1 \leq j_1$, $- j_1 \leq n_1 \leq j_1$, $- j_2 \leq m_2 \leq j_2$, $- j_2 \leq n_2 \leq j_2$, 
$- j_3 \leq m_3 \leq j_3$, $- j_3 \leq n_3 \leq j_3$.

\section{State sum model}

As we have discussed it is equivalent to think of a triangulation of a point configuration or to think of the dual graph. Let $M$ be a triangulation
of a point configuration in $\mathbb{R}^3$. Let $\partial M = \Sigma$ be its boundary.\footnote{This boundary may have one component or may be composed of different
components.} Observe that the boundary is a two dimensional triangulation in the sense that it is composed of 0-simplices, 1-simplices and 2-simplices.  

 \begin{definition}
 A colouring of $M$ is a map $c: E(M) \rightarrow \widehat{SU(2)}$, where $E(M)$ is the set of 1-simplices of $M$ and $\widehat{SU(2)}$ is the set
 of equivalent irreducible representations of $SU(2)$.
 \end{definition}

\begin{definition}
Such colouring $c$ is called admissible if for each 2-simplex there is an intertwining operator $T : H_1 \otimes H_2 \otimes H_3 \rightarrow \mathbb{C}$. 
\end{definition}
Consider now the tensor product of three irreducible representations of $SU(2)$,  $( H^{j_{12}} \otimes H^{j_{23}} \otimes H^{j_{34}}, \rho^{j_{12}} \otimes \rho^{j_{23}} \otimes \rho^{j_{34}})$.
We know from section 2.1 that the tensor product of two irreducible representations decomposes in a direct sum of irreducible ones. 
Therefore, now have two possible homomorphisms 

\begin{equation}
F_1 = H^{j_{14}} \longrightarrow ( H^{j_{12}} \otimes H^{j_{23}} ) \otimes H^{j_{34}} \nonumber
\end{equation}   
and

\begin{equation}
F_2 = H^{j_{14}} \longrightarrow H^{j_{12}} \otimes ( H^{j_{23}}  \otimes H^{j_{34}} )  \nonumber
\end{equation}   
As 

\begin{equation}
H^{j_{12}} \otimes H^{j_{23}} =  \bigoplus_{j_{13}} H^{j_{13}} \nonumber
\end{equation}  
where $j_{13}$ runs from $\mid j_{23} - j_{12} \mid$ to $j_{23} + j_{12}$. We also have that  

\begin{equation}
H^{j_{23}} \otimes H^{j_{34}} =  \bigoplus_{j_{24}} H^{j_{24}} \nonumber
\end{equation} 
where $j_{24}$ runs from $\mid j_{34} - j_{23} \mid$ to $j_{34} + j_{23}$.

\begin{definition}
The coefficients which appear in the change of basis in the diagram

\bigskip

\bigskip

\begin{tikzcd}
H^{j_{14}} \arrow[r, "F_2"]\arrow[d,  "Id"red]
& H^{j_{12}} \otimes ( H^{j_{23}} \otimes H^{j_{34}} )  \arrow[d, "" ] \\
H^{j_{14}}  \arrow[r,  "F_1" red]
& ( H^{j_{12}} \otimes  H^{j_{23}} ) \otimes H^{j_{34}}  
\end{tikzcd}

\
\bigskip

\bigskip

are called the 6j-symbols and are denoted by

\begin{equation*}
\left \{
\begin{array}{ccc}
        j_{12} & j_{23} & j_{13}   \\
        j_{34} & j_{14} & j_{24}  
        \end{array} \right \}
 \end{equation*}

\end{definition}

\bigskip 

\bigskip

Explicitly the $6j$-symbols are given by 

\begin{IEEEeqnarray}{rCl}
& \left \{
\begin{array}{ccc}
        j_{1} & j_{2} & j_{3}   \\
        j_{4} & j_{5} & j_{6}  
    \end{array}      \right \} = \nonumber \\ & \sum_{m_1, ..... , m_6}  (-1)^{\sum_{i=1}^{6} (j_{i} - m_{i})}  \  
     \widetilde{C}(j_1, - m_1, j_2 , - m_2 , j_3 , - m_3) \ 
    \widetilde{C}(j_1, m_1, j_5 , - m_5 , j_6 , m_6)  \nonumber \\
    & \times \widetilde{C}(j_4, m_4, j_2 , m_2 , j_6 , - m_6) \ \widetilde{C}(j_3, m_3, j_4 , - m_4 , j_5 , m_5) \nonumber
 \end{IEEEeqnarray}
 The $6j$-symbol has the following symmetries 
 
 \begin{IEEEeqnarray}{rCl}
\left \{
\begin{array}{ccc}
        j_{1} & j_{2} & j_{3}   \\
        j_{4} & j_{5} & j_{6}  
    \end{array}      \right \} = 
\left \{
\begin{array}{ccc}
        j_{2} & j_{1} & j_{3}   \\
        j_{5} & j_{4} & j_{6}  
    \end{array}      \right \}  =    
\left \{
\begin{array}{ccc}
        j_{3} & j_{2} & j_{1}   \\
        j_{6} & j_{5} & j_{4}  
    \end{array}      \right \}    = 
\left \{
\begin{array}{ccc}
        j_{4} & j_{2} & j_{6}   \\
        j_{1} & j_{5} & j_{3}  
    \end{array}      \right \}     
    \nonumber
 \end{IEEEeqnarray}

\begin{definition}
Let $M$ be a triangulation
of a point configuration in $\mathbb{R}^3$ and $\partial M = \Sigma$ be its boundary. Let $E^{\circ} = \{ E(M) \in M^{\circ}\}$ be the interior edges of $M$. 
The state sum is defined by 

\begin{equation}
Z_M [ \Sigma] = \sum_{c} \ \prod_{E^{\circ}} (-1)^{2j} (2j+1) \ \prod_{\triangle} (-1)^{\sum_{i=1}^{6} j_{i}}  \
\left \{
\begin{array}{ccc}
        j_{1} & j_{2} & j_{3}   \\
        j_{4} & j_{5} & j_{6}  
    \end{array}      \right \}     \nonumber
\end{equation}
where the sum is over all functions $c: E^{\circ} \rightarrow \widehat{SU(2)}$ of admissible colourings 
of the interior 1-simplices $E^{\circ}$, and $\triangle$ is the set of 3-simplices of $M$. 
\end{definition}
The state sum $Z_M[\Sigma]$ is a function of the colouring of the boundary $\Sigma$ of $M$
The colouring of the boundary is fixed
since the state sum is only over the colouring of the interior 1-simplices. And if our triangulation configuration $M$ has no boundary
then $Z_M [\Sigma] \in \mathbb{C}$.

It is known that this state sum is interpreted geometrically 
as a sum over metrics on the interior of a manifold which match the fixed boundary metric. Moreover, 
for most triangulations this state sum diverges and therefore regularisation of the sum is needed. 
We will not be concerned with the metric problem here but in a future work. Neither will we be concerned with
the regularisation issue in this work as it is a different topic which has been extensively studied, see for example 
\cite{tv}, \cite{bng}. These subjects will be treated in a different one. 
 
 Physically the state sum is a Feynman path integral of general relativity. In this case of three dimensional Euclidean general relativity. This state sum is known as the Ponzano-Regge model
 \cite{pr} and its connection to loop quantum gravity was pointed out in \cite{r2}.
 We will deal with the regularisation and the mathematics of this state sum in another work and will describe its properties.

 \section{Discussion}
 
 We have introduced a well known state sum of three dimensional loop quantum gravity. The model is a quantisation of three dimensional Euclidean general relativity without
 cosmological constant. Originally this model came from a discretisation of a Feynman path integral of three dimensional general relativity. But now, in the
 loop quantum gravity this discrete model is not something we impose for simplicity problems but it is a real consequence of the theory. 
 Space is really quantised. 
 
 This three dimensional model is one of the simplest models of loop quantum gravity since it is a toy model in the physical sense. Moreover, we have introduced mathematical background 
 which is thoroughly used in many topics of loop quantum gravity.  
 
We have described the model in a mathematical way as much as formal as we could. Our purpose is to introduce loop quantum gravity 
to mathematicians or it could also be thought as a project in which we are introducing an approach to the field of loop quantum gravity
from a mathematical side; from the point of view of a mathematician. 
 
 We hope that the present work, article \cite{gi}, and future writings from a mathematical point of view inspire mathematics and physics students as well as researchers in the field
 to approach the subject of loop quantum gravity in a more rigorous mathematical way.

\end{document}